\documentclass[12pt]{article}
\usepackage[longnamesfirst]{natbib}
\usepackage[hmargin=2.35cm, vmargin=2.54cm]{geometry}
\usepackage{tikz}
\usepackage{amsmath,amssymb,bm}
\usepackage{textgreek}
\usepackage{fancybox}
\usepackage{bbding}
\usepackage{multirow}
\usepackage{url}
\usepackage{bbold}
\usepackage{relsize}
\usepackage{subfigure}
\usepackage{array}
\usepackage[labelfont=bf]{caption}
\usepackage{amsfonts}
\usepackage{bbm}
\usepackage{dsfont}
\usepackage{amsthm}
\usepackage{amssymb}
\usepackage{enumerate}
\usepackage{graphicx}
\usepackage{lscape}
\usepackage{longtable}
\usepackage{rotating}
\usepackage{setspace}
\usepackage{dcolumn}
\usepackage{booktabs}
\usepackage{upgreek}
\usepackage{mathtools}

\usepackage{comment} 
\excludecomment{confidential}  

\usepackage{kantlipsum}
\allowdisplaybreaks

\usepackage[shortlabels]{enumitem} 
\usepackage{siunitx}
\usepackage{empheq}
\usepackage{cases} 
\usepackage[compact]{titlesec}

\titleformat{\section}{\Large\bfseries\fontsize{18}{15}}{\thesection}{0.81em}{}
\titleformat{\subsection}{\large\bfseries\fontsize{12}{10}}{\thesubsection}{0.69em}{}
\titleformat{\subsubsection}{\normalfont\bfseries\fontsize{12}{10}}{\thesubsubsection}{0.69em}{}

\usepackage{pgfplots}
\usepackage{tikz}
\usetikzlibrary{arrows,decorations.markings}

\usepackage{color}
\usepackage[bookmarks,
breaklinks=true,
	colorlinks=true,
	urlcolor=black,
	citecolor=blue,
	linkcolor=blue]{hyperref}
	
\usepackage{cleveref}
\usepackage{lineno}

\usepackage{changepage}   

\definecolor{navy}{rgb}{0,0,.5}
\definecolor{mygreen}{rgb}{0,.35,0}
\definecolor{darkgreen}{rgb}{0,.35,0}
\definecolor{mydg}{rgb}{0,.35,0}   

\newcolumntype{d}[1]{D{.}{.}{#1}}

\newcolumntype{L}{D{.}{.}{3,2}}
\newcolumntype{C}[1]{D{.}{.}{#1}}

\crefname{proposition}{Proposition}{Propositions}
\Crefname{proposition}{Proposition}{Propositions}
\crefname{claim}{Claim}{Claims}
\Crefname{claim}{Claim}{Claims}
\crefname{table}{Table}{Tables}
\Crefname{table}{Table}{Tables}
\crefname{figure}{Figure}{Figures}
\Crefname{figure}{Figure}{Figures}
\crefname{section}{Section}{Sections}
\Crefname{section}{Section}{Sections}
\crefname{appendix}{Appendix}{Appendices}
\Crefname{appendix}{Appendix}{Appendices}
\crefname{equation}{equation}{equations}
\Crefname{equation}{Equation}{Equations}
\creflabelformat{equation}{#2\textup{(#1)}#3}

\crefname{sfacts}{stylized fact}{stylized facts}
\Crefname{sfacts}{Stylized facts}{Stylized facts}
\crefname{remark}{Remark}{Remarks}
\Crefname{remark}{Remark}{Remarks}

\def\myleftmargin{2cm}

\newenvironment{enumerate0}{
\begin{enumerate}
  \setlength{\itemsep}{0pt}
  \setlength{\parskip}{0pt}
  \setlength{\parsep}{0pt}
}{\end{enumerate}}

\newenvironment{itemize0}{
\begin{itemize}
  \setlength{\itemsep}{0pt}
  \setlength{\parskip}{0pt}
  \setlength{\parsep}{0pt}
}{\end{itemize}}

\newcommand{\ie}[0]{i.e., }

\makeatletter
\newcommand{\shorteq}{%
  \settowidth{\@tempdima}{-}
  \resizebox{\@tempdima}{\height}{=}%
}

\makeatother

\newtheorem{claim}{Claim}

\long\def\symbolfootnote[#1]#2{\begingroup%
\def\thefootnote{\fnsymbol{footnote}}\footnote[#1]{#2}\endgroup}

\setpagewiselinenumbers

\begin{document}

\title{Labor Market Impact on Homelessness: Evidence from Canadian Administrative Data on Shelter Usage
}

\author{
Damba Lkhagvasuren\thanks{Dept. of Economics, Concordia University and CIREQ; e-mail: \texttt{damba.lkhagvasuren@concordia.ca}.}
\and
Purevdorj Tuvaandorj\thanks{Dept. of Economics, York University; e-mail: \texttt{tpujee@yorku.ca}.} 
}

\maketitle
\thispagestyle{empty} \setcounter{page}{0}

\begin{abstract} 
The overwhelming majority of homeless individuals are jobless, despite many expressing a willingness to work. While this strong individual-level link between homelessness and unemployment is well-documented, the broader impact of labor market dynamics on homelessness remains largely unexplored. To fill this gap, this paper investigates the impact of local labor market conditions on the duration of homelessness, using individuals' homeless shelter usage records as a proxy for measuring their homelessness duration. Specifically, drawing on Canada's National Homelessness Information System data from 2014 to 2017, we analyze how local employment growth and changes in the local employment rate affect shelter usage duration. Our findings reveal that a 1\% increase in local employment is associated with a 0.11-quarter (approximately 0.33-month) reduction in the average duration of shelter usage, while a 1\% rise in the local employment rate leads to a 0.34-quarter (approximately 1.02-month) reduction. These changes correspond to decreases of 2.9\% and 8.9\%, respectively, in the average duration of shelter stays. The findings underscore the critical role of employment opportunities in reducing homelessness and lend support to job-oriented policy interventions for the homeless. In addition, the results suggest that demographic disparities---particularly the overrepresentation of Indigenous people and men among the homeless---are partially explained by slower exit rates from homelessness within these groups.
\end{abstract}


\onehalfspacing
\clearpage

\section{Introduction}\noindent 
Homeless individuals form a diverse population, exhibiting substantial variation in their pathways into and out of homelessness, as well as in their lived experiences \!---\! even among those with seemingly similar backgrounds.
Yet, despite this diversity, a critical commonality unites the overwhelming majority of the homeless: joblessness. This strong connection between homelessness and joblessness is well-documented in the literature. For instance, in a 2009 survey conducted in Sacramento, California, \citet{acuna2009_employment} found that 91.7\% of homeless individuals were not employed. Similarly, \citet{poremski2015employment} reported an average joblessness rate of 95.9\% among homeless individuals in a large survey spanning five Canadian cities.\footnote{\citet{homeless_labour_issues} identify key labor market challenges for the homeless, such as insufficient work, inconsistent pay, and strained employer relations. \citet{shelter_rules} stress how shelter regulations limit job opportunities. \citet{homelessness_employment} highlight barriers like health issues, substance use, lack of training, and institutional factors such as discrimination and shelter policies.}  Notably, both studies reveal that a majority of homeless individuals \!---\! ranging from 64.3\% to 87.4\% \!---\! express a desire to work.

While the \emph{individual-level} link between homelessness and unemployment is well established, the impact of labor market conditions on homelessness remains unexplored. In particular, how shifts across local labor markets\,\textemdash\,such as job availability and employment growth\,\textemdash\,affect homelessness is not well understood. In the absence of such evidence, it is challenging to determine whether employment provides an effective pathway for individuals to exit homelessness, as well as to fully assess the role of policy interventions aimed at supporting employment among the homeless.

The fundamental challenge in analyzing the labor market's impact on homelessness is the difficulty of obtaining accurate empirical measurements. Two key factors contribute to this. First, homeless individuals are a hard-to-reach population, making it difficult to track changes in their numbers over time and to assess how such changes are influenced by labor market conditions. Second, major labor force surveys, such as the Current Population Survey (U.S.) and the Labour Force Survey (Canada), typically omit homelessness status, creating a significant data gap that limits our ability to study the relationship between labor market fluctuations and homelessness.

Given these empirical challenges, this study examines how local labor market conditions influence shelter usage duration\,\textemdash\,the length of time individuals stay at a homeless shelter. Shelter usage duration is used as a proxy for the overall duration of homelessness, allowing for an indirect analysis of homelessness dynamics. This approach is rooted in well-established relationships: homelessness duration is inversely related to the exit rate\,\textemdash\,the rate at which individuals transition out of homelessness\,\textemdash\,which in turn is negatively associated with the overall homelessness rate. These relationships are detailed further in \cref{fig:flows} and \cref{app:two-state}.\footnote{The notion that entry and exit rates jointly determine the size of the homeless population\,\textemdash\,and that the exit rate is inversely related to the duration of homelessness\,\textemdash\,dates back to \citet{allgood1997duration}. Another important similarity between their work and the current study is the focus on shelter usage duration derived from administrative data.} By adopting this approach, the study provides a fresh perspective on how local labor market dynamics shape homelessness outcomes.

\begin{figure}[htbp]
\caption{Homelessness Rate, Duration, and Transition Rates (Entry and Exit)}\label{fig:flows}
\hspace{-1cm}
\begin{center}
\begin{tikzpicture}[scale=0.69, font=\normalsize]

\draw [magenta, line width=1mm, ->] (-4, 1.45) to [out=40, in=180] (0.4, 3.5)  to [out=0, in=125] (3.6, 1.7);
\node at (0.4,4.05) [] {entry rate: $\lambda$};

\draw [blue, line width=1mm, ->] (4.00,-1.0) to [out=-120, in=0] (0.5, -3.6)  to [out=180, in=320] (-2.6, -2.4);
\node at (0.5,-4.1) [] {exit rate: $\theta$};

\draw [line width=0.57mm, black!70, fill=green!30] (-4.5,0) circle [radius=3]; 
\node at (-4.5,0.5) {non-homeless};
\node at (-4.5,-0.3) {stock: $1-\alpha$};
\draw [line width=0.57mm, black!70, fill=yellow!50] (4, 0) circle [radius=1.65];
\node at (4,0.5)  {homeless};
\node at (4,-0.3)  {stock: $\alpha$};

\end{tikzpicture} 
\end{center}
\caption*{\footnotesize {\it Notes:}  
This figure illustrates a two-state stock-and-flow model of homelessness, with arrows indicating transitions between states over a unit period. In this model, \(\theta\) (\(0 < \theta \leq 1\)) signifies the \emph{exit rate}, representing the proportion of the homeless population that exits homelessness each period. Conversely, \(\lambda\) (\(0 < \lambda \leq 1\)) denotes the \emph{entry rate}, the proportion of the non-homeless population that becomes homeless within the same timeframe. Both rates are strictly positive. Crucially, the model predicts that a higher exit rate (\(\theta\)) reduces both the average \emph{homelessness duration}, \(\overline{D} = \frac{1}{\theta}\), and the stationary \emph{homelessness rate}, \(\alpha = \frac{\lambda}{\lambda+\theta}\). For a detailed analytical derivation of these relationships, refer to \cref{app:two-state}. 
}
\end{figure}

The main data used in this paper are drawn from Canada's National Homelessness Information System (NHIS), an administrative system designed to collect standardized, longitudinal data on homeless shelter usage. The NHIS provides real-time records of individuals' shelter stays, covering a wide range of shelters across all provinces and territories in Canada. 
For instance, in 2014, the NHIS recorded 67,885 individuals using at least one of the 273 emergency shelters spread across all ten provinces and three territories. This comprehensive coverage, coupled with the fact that most homeless individuals in Canada utilize shelters,\footnote{\citet{gaetz2014state} estimate that over 235,000 Canadians experienced homelessness in 2014, while \citet{PiT2016}  report that 136,865 individuals used emergency shelters that same year, suggesting that 58.2\% of the homeless population accessed shelters. Furthermore, \citet{ESDC:everyone} reports that one-quarter of currently homeless individuals had not used a shelter in the past year. However, this statistic excludes those who used shelters within the past year but have since exited homelessness. As a result, the proportion of currently homeless individuals who rely on shelters could be significantly higher than three-quarters.} provides a robust foundation for identifying empirical patterns in shelter usage and uncovering critical insights into the individual-level dynamics of homelessness relative to local labor market conditions.

A central aspect of this analysis is its focus on individuals who first entered the NHIS in 2014\,\textemdash\,first-time shelter users whose unique IDs appear in the 2014 records but not in any prior NHIS data. Restricting the sample to a single entry year serves three key purposes. First, it mitigates composition biases that could arise from comparing individuals who entered the shelter system at different times under varying external conditions. Second, it ensures that individuals within the same region and shelter usage duration experience homelessness under comparable external conditions, including macroeconomic factors, social policies, and the availability and nature of street drugs. Third, it ensures a more balanced representation of demographic groups in the sample, particularly those with varying homelessness exit rates, preventing the underrepresentation of individuals who exit more quickly. This allows for a more accurate assessment of labor market effects. As discussed below, these exit rates differ substantially across certain demographic groups.

We track first-time shelter users over a four-year period (2014-2017) to examine whether their shelter usage duration responds to local labor market conditions. The analysis focuses on two key indicators of local labor market conditions: employment growth and employment rate growth. These measures are grounded in prior research on labor market dynamics. \citet{Blanchard/Katz:1992} identify local employment growth as a strong indicator of regional labor market performance based on their influential study of U.S. labor markets. \citet{fraser_institute} further argue that, amid declining labor force participation rates in Canada, the employment rate has become an increasingly reliable metric for assessing labor market conditions. Local labor markets are defined at the provincial and territorial level, as these represent the smallest administrative geographic units in Canada for which consistent labor market indicators are available \citep{local_labor_data:2021}. This definition ensures that the local labor market indicators are measured in a standardized and comparable manner across regions.

In addition to analyzing the local labor market's effect on average homelessness duration, our study also investigates its impact on chronic homelessness. Chronic homelessness refers to prolonged periods without stable housing, typically lasting six months or more, or involving repeated episodes over several years.\footnote{A comprehensive typology of homelessness is provided by \citet{Canadian_definition}.} This persistent condition severely impacts the well-being of those affected and places significant strain on public services and healthcare systems. To evaluate the role of labor markets in long-term homelessness, this study defines chronic homelessness based on shelter usage records and incorporates this definition into the analysis to inform targeted policy interventions.

Accurately tracking homelessness trends is essential for effective policymaking and resource allocation, yet the absence of reliable time-series data on homelessness presents a significant challenge. Our analysis provides a key insight: regional employment growth has substantial predictive power for homelessness trends. By incorporating the local labor market effects identified in this study, policymakers can better anticipate changes in homelessness duration and implement more targeted interventions.

Beyond its implications for employment-related policies and homelessness forecasting, the analysis reveals significant demographic disparities in shelter exits. The findings show that the overrepresentation of men and Indigenous people in the homeless population is partly attributable to their extended durations of homelessness and lower exit rates.

The remainder of the paper is organized as follows: 
\cref{sec:data} describes the NHIS sample and introduces the key variables used in the empirical analysis. \cref{sec:labormarket} discusses the local labor market indicators, focusing on the local employment growth and the local employment rate growth. 
\cref{sec:findings} outlines the empirical models and presents the main findings. 
Finally, \cref{sec:conclusion} provides concluding remarks.

\section{NHIS data}\label{sec:data} 

The primary data source for this study is Canada's National Homelessness Information System (NHIS), an administrative database designed to track standardized, longitudinal records of homeless shelter usage nationwide. The NHIS provides real-time shelter stay records, capturing essential details on individuals' check-in and check-out dates, as well as the geographic locations of the shelters.

In addition to tracking shelter stays, the NHIS records key demographic information, such as age, gender, Indigenous identity, family status (i.e., whether the individual is accompanied by a dependent or family member), citizenship, and veteran status, which identifies former members of the Armed Forces or the Royal Canadian Mounted Police.

\subsection{Sample description}\label{sec:sample}
In 2014, a total of 67,885 individuals across Canada's 10 provinces and three territories used at least one of the 273 emergency shelters covered by the National Homelessness Information System (NHIS). Among them, 31,588 were first-time shelter users aged 20 and older. This study focuses on these individuals. A shelter user is classified as a first-time user if their unique ID appears in the NHIS for the first time in 2014 and does not appear in any prior records. To analyze their shelter usage patterns, we track these individuals over a four-year period from January 1, 2014, to December 31, 2017.\footnote{The sample period for this study predates the COVID-19 pandemic. Future research could benefit from comparing pre-COVID and COVID-era homelessness rates to gain insights into the potential impact of the pandemic on shelter usage durations.} 

As outlined in the Introduction, restricting the sample to first-time shelter users within a single year serves three key purposes. First, it mitigates composition biases that could arise from comparing individuals who entered the shelter system under different external conditions. Second, it ensures that individuals with the same shelter usage duration within a given region experienced homelessness under comparable macroeconomic and policy environments. Third, it prevents the underrepresentation of demographic groups with faster homelessness exit rates, enabling a more precise assessment of labor market effects.

\subsection{Demographic characteristics of the sample}
The average age of individuals in the sample is 39.82 years. \Cref{tab:composition} presents the demographic composition of the homeless population based on available categorical variables. The key patterns emerging from the sample are as follows: 
\begin{itemize0} 
\item[$(i)$] Men constitute approximately 70\% of the sample, highlighting the predominance of men in the homeless population.
\item[$(ii)$] More than three-quarters of homeless individuals are unaccompanied by dependents or family members.
\item[$(iii)$] Individuals identifying as Indigenous make up over 28\% of the sample\,\textemdash\,nearly five times their share of the overall adult population in Canada, which is approximately 5\% \citep{statcanada:2021}. This underscores the significant overrepresentation of Indigenous people among the homeless.
\item[$(iv)$] Non-citizens account for about 5\% of the sample, while veterans represent less than 2\%. However, these figures should be interpreted with caution, as citizenship status is unknown for 42\% of individuals, and veteran status is unknown for 58\%.
\end{itemize0}

\renewcommand{\arraystretch}{0.85}
\begin{table}[htbp]
\begin{center}
\caption{Demographic Composition of the NHIS Sample}\label{tab:composition}
\begin{tabular}{llcrllr}
\toprule \\[-3mm]
\multicolumn{2}{l}{\emph{gender}}	&  & \qquad \qquad & \multicolumn{2}{l}{\emph{citizenship}}	&		\\	
&	female	&	30.04\%	&   &   \qquad \qquad &	Canadian	&	53.27\%	\\
&	male	&	69.96\%  & &	&					non-Canadian	\quad &	4.61\%	\\
\qquad \qquad &		&	&	&&				unknown	&	42.12\%	\\
\multicolumn{3}{l}{\emph{family status}}	& & 		&		& \\
&	accompanied &	11.09\%	&   & \multicolumn{2}{l}{\emph{veteran status}}	&	\\
&	unaccompanied  \qquad	&	77.12\% && &				veteran	&	1.50\%	\\
&	unknown	&	11.79\%	&    & &				non-veteran	&	40.06\%	\\
&		&		&	& &						unknown &	58.44\%	\\
\multicolumn{2}{l}{\emph{Indigenous status}} &&&&			\\
&	Indigenous	&	28.28\% &	& &					\\
&	non-Indigenous	&	46.49\%	 &	& &					\\
&	unknown	&	25.23\%	 &    & &					
\\[2mm]
\bottomrule
\end{tabular}
\vspace{0.1cm}
\caption*{\footnotesize {\it Notes:} 
The average age is 39.82, with a standard deviation of 13.27. Here, ``accompanied" refers to individuals being accompanied by a dependent or family member.
}
\end{center}
\end{table}

\subsection{Duration of homelessness}\label{sec:emp_duration}
The main variable of interest is the duration of homelessness. As previously mentioned, in this analysis, shelter usage duration serves as a proxy for homelessness duration. For brevity, and with a slight concession in terminology, we refer to shelter usage duration simply as homelessness duration. Throughout the analysis, $D$ denotes the duration of homelessness.

However, it is crucial to acknowledge the distinction between homelessness duration and shelter usage duration: a brief absence from shelter records does not necessarily imply that an individual has exited homelessness. For example, individuals may alternate between staying in shelters and living on the streets during a single episode of homelessness. Given this mobility across different living situations, daily or weekly measures of shelter usage may not fully capture the extent of an individual's homelessness experience. To address this, we calculate the duration of homelessness at a quarterly frequency (\ie 90 days). Shelter stays separated by less than a quarter are therefore treated as part of the same episode of homelessness. Furthermore, measuring homelessness duration on a quarterly basis aligns well with the definition of chronic homelessness, as the reference periods for defining chronic homelessness are typically integer multiples of quarters (e.g., 2 or 12 quarters), as outlined below.

The summary statistics for homelessness duration, $D$, are presented in \cref{tab:duration_summary}. The average duration is 3.758 quarters, slightly less than one year. The standard deviation is considerable, indicating significant variation across individuals.

\begin{table}[htpb]
\begin{center}
\caption{Homelessness Duration and Chronic Homelessness}
\label{tab:duration_summary}
\begin{tabular}{cclSS}
\toprule 
\multicolumn{1}{c}{\multirow{2}{*}{\it variables}} & \; \qquad & \multicolumn{1}{c}{\multirow{2}{*}{\it description}} & \multicolumn{1}{c}{\multirow{2}{*}{\it mean}} &  \multicolumn{1}{c}{\it standard} \\
&&&& {\it deviation} \\
\midrule \\[-3mm]
$D$ && duration of homelessness (in quarters) \qquad \qquad  &  3.757894 &   4.294373 \\ [1mm]
 $I_1$ && first criterion of chronic homelessness & 0.051663  &   0.221349 \\ [1mm]
 $I_2$ &&  second criterion of chronic homelessness & 0.0968305 &   0.2957312\\ [1mm]
  $\chi$ && chronic homelessness & 0.1270223  &   0.333003\\[2mm]
\bottomrule 
\end{tabular}
\vspace{0.1cm}
\end{center}
\end{table}%

\subsection{Chronic homelessness}\label{sec:chron}
In addition to homelessness duration, we also incorporate a dichotomous variable representing chronic homelessness. While the specific definition of chronic homelessness varies across countries and regions, it typically involves two key criteria: continuous homelessness for an extended period, usually between 6 and 12 months, \emph{and} multiple episodes of homelessness within a specified timeframe, typically over three years. Based on these criteria, we introduce two dummy variables, $I_1$ and $I_2$, defined as follows. 
\begin{itemize0}
\item $I_1$ takes a value of 1 if the individual experienced at least 180 days of homelessness within a 365-day period, and 0 otherwise. 
\item $I_2$ takes a value of 1 if the individual experienced at least one episode of homelessness per year for three consecutive years, and 0 otherwise. 
\end{itemize0}
 
Any individual satisfying either of the conditions associated with $I_1$ and $I_2$ is considered chronically homeless. Specifically, a dummy variable indicating chronic homelessness is defined as follows:  
\begin{equation}
\chi = \begin{cases}
1 & \text{if } I_1+I_2 \ge 1, \\
0 & \text{otherwise.}
\end{cases}
\end{equation}
Summary statistics for these variables are presented in \cref{tab:duration_summary}.
It shows that among individuals who used shelters for the first time in 2014, 12.7\% experienced chronic homelessness.

\section{Quantifying local labor market conditions}\label{sec:labormarket}
The aim of this paper is to examine whether labor market conditions affect the duration of homelessness and chronic homelessness. 

\subsection{Two indicators}
To capture local labor market conditions, we consider two key indicators:  
\begin{enumerate0}
    \item[($i$)] local employment growth, and
    \item[($ii$)] local employment rate growth.
\end{enumerate0}
Both indicators are constructed by controlling for aggregate effects, as explained below. 

The selection of these indicators is guided by the following considerations. First, \citeauthor{Blanchard/Katz:1992}'s (\citeyear{Blanchard/Katz:1992}) seminal analysis of regional labor market dynamics across U.S. states established local employment growth as a key measure of regional labor market performance. Second, \citet{fraser_institute} argue that, given Canada's declining labor force participation rates, changes in the employment rate provide a more reliable indicator of labor market conditions. Third, beyond labor force participation trends, employment rate growth is particularly relevant in high-immigration countries like Canada, as it inherently adjusts for demographic shifts and provides a more accurate measure of job availability relative to the working-age population.

\subsection{Data source and local labor markets}

The labor market indicators used in this study are constructed from annual labor force series at the provincial and territorial level. These series are tabulated by \citet{local_labor_data:2021} using the Canadian Labour Force Survey, which provides consistent labor market data across subnational regions. Detailed definitions of the survey variables are available in Statistics Canada's \emph{Guide to the Labour Force Survey}.\footnote{URL: \url{https://www150.statcan.gc.ca/n1/pub/71-543-g/71-543-g2018001-eng.htm}}

While provinces and territories represent the smallest administrative units for which consistent labor market data are available, treating each as a separate labor market can reduce the accuracy of indicators for regions with very small labor forces, as the reliability of labor market performance measures depends on labor force size. To address this, we merge certain smaller provinces and territories into larger, geographically proximate regions to improve measurement accuracy. Specifically, Nunavut and Yukon\,\textemdash\,whose combined labor force accounted for only 0.18\% of the Canadian total in 2014\,\textemdash\,are merged with the Northwest Territories. Similarly, Prince Edward Island, which represented just 0.43\% of the national labor force in 2014, is combined with neighboring Nova Scotia. These adjustments yield a total of 10 regions, ensuring a more accurate representation of local labor market conditions. With this regional framework in place, we now formally define the two labor market indicators.

\subsection{Local employment growth}
Following \citet{Blanchard/Katz:1992}, local employment growth is constructed as the change in the log ratio of local employment to aggregate (\ie national) employment. To construct this indicator, we first define the relative employment of region $j$ in year $t$ as follows:
\begin{equation}
e_{j,t} = \ln(E_{j,t}) - \ln(\overline{E}_{t}),
\end{equation} 
where $E_{j,t}$ is the total employment of region $j$, and $\overline{E}_{t}$ is the aggregate employment of Canada in year $t$. The annual growth of this relative employment is given by 
\(
\Delta e_{j,t} = e_{j,t+1} - e_{j,t}.
\)
Thus, the indicator of local employment growth over the sample period is defined as:
\begin{equation}
\Delta \mathcal{E}_j = \frac{1}{3} \sum_{t=2014}^{2016} \Delta e_{j,t}.
\end{equation}
This indicator represents the percentage change in regional employment relative to the national average.

\subsection{Local employment rate growth} 
Let $r_{j,t}$ denote the employment rate of region $j$ relative to the aggregate employment rate of Canada in year $t$, \ie  
\begin{equation}
r_{j,t}= \frac{R_{j,t}}{\overline{R}_{t}},
\end{equation}
where $R_{j,t}$ is the employment rate of region $j$ and $\overline{R}_{t}$ is the aggregate (national) employment rate in year $t$.  The annual growth rate of the relative employment rate is then given by 
\(
\Delta r_{j,t} = r_{j,t+1} - r_{j,t}.
\)
The average growth of the (relative) employment rate for region $j$ over the sample period is thus defined as
\begin{equation}\label{eq:define_ER}
\Delta \mathcal{R}_j = \frac{1}{3} \sum_{t=2014}^{2016} \Delta r_{j,t}.
\end{equation}

\subsection{Summary statistics}
\Cref{tab:local} presents key statistics for the labor market indicators. The mean and standard deviation of local employment growth are -0.0063 and 0.0108, respectively, while those of local employment rate growth are -0.0039 and 0.0076. \cref{tab:local} also shows a strong and significant positive correlation between the two indicators, which is unsurprising, as both are based on the common underlying variable of local employment.

The means of both indicators are close to zero, as expected, since they were constructed by controlling for aggregate effects. Additionally, the stability of Canada's aggregate labor market indicators during the sample period further contributes to this result. However, for the purposes of this analysis, the means are not particularly important. What matters is the magnitude of the standard deviations, as they capture regional differences in labor market shifts, which are central to the analysis. In percentage terms, the standard deviations of the two indicators across regions are 1.08\% and 0.76\%, respectively. Although these variations may seem small, we show below that these small differences have significant effects on both homelessness duration and chronic homelessness.

\renewcommand{\arraystretch}{1.00}
\begin{table}[htbp]
\begin{center}
\caption{Local Labor Market Indicators}\label{tab:local}
\sisetup{
  round-mode               = places,
  round-precision          = 4,
  group-digits             = false,
  input-open-uncertainty   = ,
  input-close-uncertainty  = ,
  table-format             = +2.4),
  minimum-decimal-digits   = 4,
  table-align-text-before  = false
}
\begin{tabular}{l c c c }
\toprule
	& \; \qquad {\it local}  & \; \qquad & \; \qquad {\it local}    \\
		& \; \qquad {\it employment}  & \; \qquad & \; \qquad {\it employment}    \\
	& \; \qquad {\it growth, $\Delta \mathcal{E}_j$} && \; \qquad {\it rate growth, $\Delta \mathcal{R}_j$}  \\
	\midrule \\[-4mm] 
mean			& \multicolumn{1}{L}{-0.0063}      	&&  \multicolumn{1}{L}{-0.0039}	    \\
standard deviation	&  \multicolumn{1}{L}{0.0108}  	&&  \multicolumn{1}{L}{0.0076}	     \\ [1mm]
\midrule \\ [-3mm]
Pearson correlation, $\rho^{\text{P}}_{\Delta \mathcal{E}_j, \Delta \mathcal{R}_j} $  	& \multicolumn{3}{S}{0.8047***}  \\
	&   \multicolumn{3}{S}{(0.0050)}         \\ [1mm]
Spearman rank correlation, $\rho^{\text{S}}_{\Delta \mathcal{E}_j, \Delta \mathcal{R}_j} $ 	& \multicolumn{3}{S}{0.7333**}  \\
	&   \multicolumn{3}{S}{(0.0158)}          \\ [1mm]
\bottomrule
\end{tabular}
\vspace{0.1cm}
\caption*{\footnotesize {\it Notes:} The p-values are in parentheses. Significance levels are marked as: ** $\text{p}<0.05$,  *** $\text{p}<0.01$}
\end{center}
\end{table}

\section{Empirical findings}\label{sec:findings}
To evaluate the impact of local labor markets, we employ a linear regression model for homelessness duration ($D$) and a logistic regression model for the dichotomous variable ($\chi$), which represents chronic homelessness. In each model, we include one labor market indicator at a time, along with the individual's age and other individual-level characteristics presented in \cref{tab:composition}. Each observation in the regressions represents an individual who began using homeless shelters in 2014 (see \cref{sec:sample}). 

 \subsection{Main results}

\cref{tab:empgrowth} presents the impact of local employment growth. For both specifications, the coefficient for employment growth, $\Delta \mathcal{E}_j$, is negative and significant at the 0.01 level. The first column of \cref{tab:empgrowth} indicates that a 1\% increase in local employment is associated with a 0.11-quarter reduction in the average duration of homelessness. This suggests that individuals in regions with higher employment growth tend to experience shorter durations of homelessness. Chronic homelessness ($\chi$) also responds negatively to local employment growth, indicating that the likelihood of chronic homelessness decreases as local employment increases.

 \begin{table}[htbp]
 \medskip
\begin{center}
\caption{Impact of Local Employment Growth}\label{tab:empgrowth}
\sisetup{
  round-mode               = places,
  round-precision          = 3,
  group-digits             = false,
  input-open-uncertainty   = ,
  input-close-uncertainty  = ,
  table-format             = +1.3),
  minimum-decimal-digits   = 3,
  table-align-text-before  = false
}
\begin{tabular}{lSSSS} 
\toprule 
\\[-4mm]
&   	{OLS regression of}			&  	{\quad }	&	{Logit regression of}		\\  
&   	{homelessness duration,}			&  	{\quad }	&	{chronic homelessness,}		\\  
&   	{$D$}			&  	{\quad }	&	{$\chi$}		\\ [1mm] 
\cmidrule[0.3mm](lr){2-2} \cmidrule[0.3mm](lr){4-4}
\\[-4mm]		
age						&	-0.001		&		&	0.010***		\\
						&	(0.002)		&		&	(0.001)		\\ [2mm]
female					&	-0.317***		&		&	-0.016		\\
						&	(0.046)		&		&	(0.044)		\\ [2mm]
accompanied				&	-0.704***		&		&	-0.420***		\\
						&	(0.068)		&		&	(0.077)		\\ [2mm]
family situation unknown		&	1.177***		&		&	0.384***		\\
						&	(0.082)		&		&	(0.066)		\\ [2mm]
Indigenous				&	0.471***		&		&	0.315***		\\
						&	(0.052)		&		&	(0.050)		\\ [2mm]
Indigenous  status unknown	&	-0.507***		&		&	0.357***		\\
						&	(0.069)		&		&	(0.063)		\\ [2mm]
{\bf employment growth,}  $\Delta \mathcal{E}_j$ 		
						&	-10.950***		&		&	-14.497***		\\
						&	(3.852)		&		&	(3.466)		\\ [2mm]			
citizenship \& veteran statuses	& {\Large $\checkmark$}  &   &  {\Large $\checkmark$}	\\ [3mm] 
\midrule
adjusted or pseudo $R^2$ 	&  0.036 & & 0.022 \\
\bottomrule
\end{tabular}
\vspace{0.1cm}
\caption*{\footnotesize {\it Notes:} 
The variable ``accompanied" is a binary indicator representing whether the person is accompanied by a dependent or family member. Checkmarks indicate the additional variables included in the regression. In each specification, the reference group consists of non-veteran, non-Indigenous, male Canadian citizens who are unaccompanied by a dependent or family member. 
Standard errors are in parentheses.
Significance levels are marked as: * $\text{p}<0.1$, ** $\text{p}<0.05$,  *** $\text{p}<0.01$.  
As noted in \cref{sec:data}, citizenship and veteran statuses are unknown for a substantial portion of the sample. Given this limitation and to conserve space, the coefficients for the multiple dummy variables representing these statuses are not reported.
}
\end{center}
\end{table}

\cref{tab:emprategrowth} shows the impact of the local employment rate. The coefficient for employment rate growth, $\Delta \mathcal{R}_j$, is negative and significant at the 0.01 level for both specifications. The first column indicates that a 1\% increase in the local employment rate leads to a 0.34-quarter reduction in the average duration of homelessness, equivalent to approximately one month. The table also demonstrates that an increase in the local employment rate is associated with significant reductions in chronic homelessness.

As noted in \cref{tab:duration_summary}, the average homelessness duration is 3.758 quarters. This number, along with the estimates in  \cref{tab:emprategrowth,tab:empgrowth}, suggests that a 1\% increase in local employment growth and local employment rate growth leads to a 2.9\% and 8.9\% reduction in the average homelessness duration, respectively.

A comparison of \cref{tab:emprategrowth,tab:empgrowth} reveals that the local employment rate growth, $\Delta \mathcal{R}_j$, 
has a significantly stronger impact, in absolute terms, than local employment, $\Delta \mathcal{E}_j$. 
Moreover, given that the employment rate is considered a more reliable indicator of labor market conditions in Canada \citep{fraser_institute}, the findings in \cref{tab:emprategrowth} may better reflect the impact of local labor markets on homelessness duration and chronic homelessness.

\renewcommand{\arraystretch}{1}
 \begin{table}[htbp]
\begin{center}
\caption{Impact of Local Employment Rate Growth}\label{tab:emprategrowth}
\sisetup{
  round-mode               = places,
  round-precision          = 3,
  group-digits             = false,
  input-open-uncertainty   = ,
  input-close-uncertainty  = ,
  table-format             = +1.3),
  minimum-decimal-digits   = 3,
  table-align-text-before  = false
}
\begin{tabular}{lSSSS} 
\toprule 
\\[-4mm]
&   	{OLS regression of}			&  	{\quad }	&	{Logit regression of}		\\  
&   	{homelessness duration,}			&  	{\quad }	&	{chronic homelessness,}		\\  
&   	{$D$}			&  	{\quad }	&	{$\chi$}		\\ [1mm] 
\cmidrule[0.3mm](lr){2-2} \cmidrule[0.3mm](lr){4-4}
\\[-4mm]	
age						&	0.000		&		&	0.010***		\\
						&	(0.002)		&		&	(0.001)		\\ [2mm]
female					&	-0.323***		&		&	-0.023		\\
						&	(0.046)		&		&	(0.044)		\\ [2mm]
accompanied				&	-0.733***		&		&	-0.448***		\\
						&	(0.068)		&		&	(0.077)		\\ [2mm]
family situation	unknown 		&	0.873***		&		&	0.064		\\
						&	(0.094)		&		&	(0.077)		\\ [2mm]
Indigenous				&	0.415***		&		&	0.254***		\\
						&	(0.053)		&		&	(0.051)		\\ [2mm]
Indigenous status unknown	&	-0.668***		&		&	0.223***		\\
						&	(0.071)		&		&	(0.065)		\\ [2mm]
{\bf employment rate growth,} $\Delta \mathcal{R}_j$
						&	-33.598***		&		&	-35.800***		\\
						&	(4.793)		&		&	(4.340)		\\ [2mm]
citizenship \& veteran statuses	& {\Large $\checkmark$}  &   &  {\Large $\checkmark$}	\\ [3mm] 
\midrule
adjusted or pseudo $R^2$		&	0.037		&		&	0.025		\\ [0mm]	
\bottomrule
\end{tabular}
\caption*{\footnotesize {\it Notes:} The same as in \cref{tab:empgrowth}.
}
\end{center}
\end{table}

\subsection{Demographic effects}
Although the primary objective of this study is to examine the impact of local labor market conditions, 
\cref{tab:empgrowth,tab:emprategrowth} also provide valuable insights into demographic effects. These effects are summarized as follows: 
\begin{itemize0} 
\item[($i$)] Female individuals experience homelessness for 0.32 quarters less, on average, compared to their male counterparts. This finding aligns with the earlier work of \citet{allgood1997duration}, which reported a similar effect using U.S. shelter usage data. 
\item[($ii$)] Individuals accompanied by a dependent or family member tend to experience shorter durations of homelessness and fewer instances of chronic homelessness.
\item[($iii$)] Individuals identifying as Indigenous experience homelessness for 0.42 to 0.47 quarters longer than otherwise similar individuals and have higher rates of chronic homelessness compared to the rest of the homeless population. 
\end{itemize0}

Our analysis offers valuable insights into two important aspects of data: ($i$) the overrepresentation of Indigenous individuals within the homeless population and ($ii$) significant gender disparities in homelessness (see \cref{tab:composition}). 
The significant negative coefficients of the variable \emph{Indigenous} in the regressions of homelessness duration suggest that the overrepresentation of Indigenous individuals in the homeless population may be partly due to their slower exit rate from homelessness, contributing to their higher overall homelessness rates (see \cref{fig:flows}). Similarly, the negative coefficients of the variable \emph{female} in the regressions of homelessness duration indicate that, on average, males experience longer durations of homelessness, reflecting a slower exit rate, which helps to explain the predominance of males in the homeless population.

\section{Conclusions}\label{sec:conclusion}

Reducing the social and economic burden of homelessness requires not only preventing individuals from becoming homeless but also, critically, minimizing the time they remain without stable housing. While the small but rapidly growing literature on homelessness offers valuable insights into the strong individual-level connection between homelessness and joblessness, it provides comparatively limited understanding of how labor market conditions shape the duration of homelessness. Addressing this gap is essential for developing effective policies that support transitions out of homelessness and foster housing stability.

This paper contributes to filling this gap by examining the relationship between local labor market conditions and the duration of homelessness shelter usage in Canada. Our findings reveal that a 1\% increase in local employment reduces the average duration of shelter usage by 0.11 quarters, while a 1\% rise in the local employment rate shortens it by 0.34 quarters, corresponding to decreases of 2.9\% and 8.9\% in the average duration of homelessness, respectively. The results further indicate that homeless individuals in regions with higher employment growth face a reduced likelihood of experiencing chronic homelessness compared to those in areas with slower employment growth.

Nevertheless, these findings warrant careful interpretation. The analysis focuses on relative durations of homelessness by comparing individuals who became homeless during the same period but resided in regions with differing labor market conditions. The results do not necessarily imply that employment growth \emph{universally} reduces homelessness. Instead, they highlight that, \emph{all else being equal}\,\textemdash\,including \emph{the rate of entry} into homelessness\,\textemdash\,individuals in regions with stronger employment growth tend to exit homelessness more quickly than those in areas with weaker labor market performance.

Our findings that local employment growth and increases in employment rates significantly reduce shelter usage duration underscore the potential of local employment policies to mitigate homelessness. 
This aligns with previous work by \citet{role_of_emp1} and \citet{poremski2015employment}, which emphasizes the pivotal role of employment in addressing the multifaceted challenges of homelessness.

The findings suggest that employment growth has significant predictive power for homelessness trends. 
Accurate forecasts of homelessness help national and local authorities and non-governmental organizations respond more effectively to changes in homelessness and allocate resources efficiently.
However, generating reliable forecasts based solely on homelessness data is challenging due to the absence of comprehensive time-series data and the complex nature of homelessness. Our results demonstrate that local labor market forecasts can serve as a reliable predictor of regional homelessness trends.\footnote{Other potentially relevant predictors that may complement local labor market forecasts include housing affordability \citep{pawson_housing_afford}, social assistance benefits \citep{falvo2024}, and population-level emergency department visits \citep{Strobel_EDvisits}.}  
By combining local employment forecasts with the observed relationship between employment growth and shelter usage duration, policymakers and organizations can anticipate changes in homelessness more effectively and implement targeted interventions.

Furthermore, although not the primary focus of this study, the analysis provides valuable insights into demographic disparities in homelessness. Specifically, the overrepresentation of Indigenous people and men in the homeless population is partly attributable to their slower exit rates from homelessness.

These findings suggest that demographic factors, alongside economic conditions, play a critical role in shaping homelessness outcomes. To fully understand these dynamics, future research should explore both entry and exit rates of homelessness, ideally using a dynamic structural framework. Such an approach would offer a more comprehensive view of the processes that lead to and prolong homelessness, helping to inform more targeted interventions.

\subsection*{Acknowledgements}
Part of the analysis in this paper was presented at the 2024 CIREQ Interdisciplinary Workshop on Homelessness in Montreal, Canada, where it benefited from the insights and feedback of the workshop participants. We also extend our gratitude to Nan Zhou, Senior Policy Analyst at Housing, Infrastructure, and Community Canada (HICC), for his insightful comments, and to Bilguun Sukhbaatar for his excellent research assistance. This work was supported by Infrastructure Canada under Grant number 4100009143 and Concordia University Research Fund under Grant number 300010818. 

\subsection*{Data availability statement}
This paper uses two datasets: ($i$) Annual labor force characteristics by province, territory, and economic region;
($ii$) Administrative data on individuals' emergency shelter usage from Canada's National Homelessness Information System. The first dataset is publicly available at Statistics Canada's website (\url{https://doi.org/10.25318/1410009001-eng}). The second dataset cannot be shared due to confidentiality requirements to protect the sensitive information of emergency shelter users and to comply with HICC's data governance policies. 

\clearpage
\bibliographystyle{chicago}
\bibliography{homeless_emp.bib}

\clearpage
\setcounter{page}{1} 
\renewcommand{\thepage}{\roman{page}} 

\appendix

\section*{Appendix} 

\section{Stock-and-flow model of homelessness}\label{app:two-state}

In this appendix, we analytically derive the homelessness rate and the average duration of homelessness implied by the stock-and-flow model depicted in \cref{fig:flows}, using its Markov chain representation.

\subsection{Markov chain}
Consider a continuum of individuals, each uniquely identified by an index \(i\) ranging from 0 to 1 (i.e., \(i \in [0,1]\)), where the total population has unit mass. Let \(h_{i,t}\) denote the homelessness status of person \(i\) at time \(t \in \{0,1,2,\ldots\}\), where 
\begin{equation*}
	h_{i,t} = 
	\begin{cases}
		1 & \text{if individual \(i\) is housed at time \(t\)}, \\
		0 & \text{if individual \(i\) is homeless at time \(t\)}.
	\end{cases}
\end{equation*}

The stock-and-flow model in \cref{fig:flows} can be represented as a Markov chain with state space \(\mathcal{S} = \{0,1\}\) and transition matrix:
\begin{equation*}\label{pi_2a}
	\mathbf{\Pi} 
	= \begin{bmatrix}
		1 - \theta & \theta \\
		\lambda & 1 - \lambda
	\end{bmatrix},
\end{equation*}
where \(\theta = \mathrm{Pr}(h_{i,t+1} = 1 \mid h_{i,t} = 0)\) is the probability of transitioning from homelessness to being housed, and \(\lambda = \mathrm{Pr}(h_{i,t+1} = 0 \mid h_{i,t} = 1)\) is the probability of transitioning from being housed to homelessness, for each individual \(i\) and time \(t\). As specified in \cref{fig:flows}, both \(\theta\) and \(\lambda\) are positive with \(0 < \theta \le 1\) and \(0 < \lambda \le 1\).\footnote{This appendix focuses on the properties of the two-state Markov chain that are essential for understanding the relationship between the homelessness rate and its duration. For additional statistical properties of the two-state Markov chain not explored in this document, refer to \citet{Hamilton}, \citet{Hamilton_book}, and \citet{BMC}.}

\subsection{Homelessness rate}
Our first task is to determine the stationary homelessness rate using the transition probabilities in \(\mathbf{\Pi}\).

\begin{claim}\label{claim1}
	In the stock-and-flow model, the stationary rate of homelessness is
	\begin{equation*}\label{ergod_rate}
		\alpha = \frac{\lambda}{\lambda + \theta}.
	\end{equation*}
\end{claim}

\begin{proof}
	Let \(\bm{\phi}\) be a column vector representing the stationary distribution of the Markov chain, expressed as \(\bm{\phi} = [\alpha\;\;1-\alpha]^\top\). Here, \(\alpha\) is the proportion of individuals who are homeless, and \(1-\alpha\) is the proportion who are housed. Since the total population is normalized to 1, \(\alpha\) corresponds to the homelessness rate, and \(1-\alpha\) to the housed rate.
	
	The stationary distribution satisfies
	\begin{equation*}\label{ergod_eq1}
		\mathbf{\Pi}^\top \bm{\phi} = \bm{\phi},
	\end{equation*}
	which can be written explicitly as
	\begin{equation}\label{ergod_eq2}
		\begin{bmatrix}
			1-\theta & \lambda \\
			\theta & 1-\lambda
		\end{bmatrix}
		\begin{bmatrix}
			\alpha \\
			1-\alpha
		\end{bmatrix}
		=
		\begin{bmatrix}
			\alpha \\
			1-\alpha
		\end{bmatrix}.
	\end{equation}
	Solving \cref{ergod_eq2} for \(\alpha\) yields
	\[
	\alpha = \frac{\lambda}{\lambda + \theta},
	\]
	completing the proof.
\end{proof}

\subsection{Homelessness duration}
Our second task is to calculate the average duration of homelessness, denoted by \(\overline{D}\).

\begin{claim}\label{claim2}
	In the stock-and-flow model, the average duration of homelessness is the inverse of the exit rate \(\theta\), namely
	\begin{equation*}\label{eq:duration}
		\overline{D} = \frac{1}{\theta}.
	\end{equation*}
\end{claim}
\begin{proof}
	Consider an individual who has just become homeless. Each period, this individual faces a probability \(\theta\) of exiting homelessness, effectively forming a sequence of Bernoulli trials. Consequently, the duration of homelessness follows a geometric distribution with success probability \(\theta\). Thus, the probability of remaining homeless for exactly \(k\) periods is
	\begin{equation*}
		\Pr(d = k) = (1-\theta)^{k-1}\,\theta, \quad k = 1,2,3,\ldots.
	\end{equation*}
	The expected duration of homelessness is therefore
	\begin{equation}\label{eq:mean}
		\overline{d} = \sum_{k=1}^\infty k \,(1-\theta)^{k-1}\,\theta.
	\end{equation}
	By expanding the right-hand side of \cref{eq:mean}, we obtain:
\begin{align*}
	\sum_{k=1}^\infty k (1-\theta)^{k-1} \theta   
	&= \theta + (1-\theta) \sum_{k=2}^\infty k (1-\theta)^{k-2} \theta \\
	&= \theta + (1-\theta) \sum_{k=2}^\infty (1+k-1 ) (1-\theta)^{k-2} \theta \\
	&= \theta + (1-\theta) \left[ \theta \sum_{k=2}^\infty (1-\theta)^{k-2}  + \sum_{k=2}^\infty (k-1) (1-\theta)^{k-2} \theta  \right].
\end{align*}
	The first and second terms inside the square brackets simplify to 1 and \(\overline{d}\), respectively. Combining these results with \cref{eq:mean} yields
	\begin{equation*}
		\overline{d} 
		= \theta + (1-\theta)(1 + \overline{d}),
	\end{equation*}
	which further simplifies to
	\begin{equation*}
		\overline{d} = \frac{1}{\theta}.
	\end{equation*}
	Since all homeless individuals share the same exit rate, the average duration of homelessness (denoted \(\overline{D}\)) is also equal to $1/\theta$. This completes the proof of \cref{claim2}.
\end{proof}

\end{document}